\documentclass[11pt,a4paper]{amsart}

\usepackage[T1]{fontenc}
\usepackage[utf8]{inputenc}
\usepackage[english]{babel}
\usepackage[a4paper,margin=2.8cm]{geometry}

\usepackage{amsmath,amssymb,amsfonts,amsthm,mathtools}
\usepackage{mathrsfs}
\usepackage{bm}

\usepackage{tikz}
\usepackage{tikz-cd}
\usetikzlibrary{arrows.meta,calc,decorations.pathmorphing,positioning}

\usepackage{enumitem}
\usepackage{xcolor}

\usepackage[colorlinks=true,linkcolor=blue,citecolor=blue,urlcolor=blue]{hyperref}

\theoremstyle{plain}
\newtheorem{theorem}{Theorem}[section]
\newtheorem{proposition}[theorem]{Proposition}
\newtheorem{lemma}[theorem]{Lemma}
\newtheorem{corollary}[theorem]{Corollary}

\theoremstyle{definition}
\newtheorem{definition}[theorem]{Definition}
\newtheorem{remark}[theorem]{Remark}
\newtheorem{example}[theorem]{Example}

\newcommand{\PC}{\mathrm{PC}}

\newcommand{\id}{\mathrm{id}}

\title[Pairwise-comparison-valued cosurfaces]{Pairwise-comparison-valued cosurfaces: a projective framework for multi-scale relational structures}

\author{Jean-Pierre Magnot}
\address{{SFR MATHSTIC, LAREMA, Universit\'e d’Angers, 2 Bd Lavoisier, 
49045 Angers cedex 1, France;  Lyc\'ee Jeanne d'Arc, 40 avenue de Grande Bretagne, 63000 Clermont-Ferrand, 
France}; Lepage Research Institute, 17 novembra 1, 081 16 Presov, Slovakia}
\email{\small magnot@math.cnrs.fr; jean-pierr.magnot@ac-clermont.fr; jp.magnot@gmail.com}

\begin{document}
\begin{abstract}
We introduce cosurfaces with values in the group \(\PC_n(H)\) of \(H\)-valued reciprocal pairwise comparison matrices. 
The composition law is covariant on upper triangular coefficients and contravariant on lower triangular coefficients, which makes \(\PC_n(H)\) a natural target for oriented gluing constructions. 
Starting from a directed family of finite oriented discretizations, we define finite configuration spaces, coarse-graining maps induced by ordered refinements, and the associated universal projective limit. 
This yields a multi-scale organization of local comparative data in which global objects are reconstructed only through compatibility across scales. 
In the stochastic setting, projectively compatible probability laws define a cylindrical semantics on the limit space. 
We also introduce inconsistency observables, interpreted as discrete curvature-type defects measuring obstructions to global coherence. 
The resulting framework is simultaneously geometric, algebraic, and probabilistic, and suggests a foundational perspective on relational structures built from local comparisons rather than absolute observables.
\end{abstract}

\maketitle
\section{Introduction}

A recurrent theme in both the foundations of physics and the theory of comparative data is the tension between \emph{local relational information} and \emph{global coherent structure}.
In many situations, one does not begin with absolute observables attached once and for all to isolated objects.
One begins instead with relative statements, local comparisons, transition data, or compatibility conditions, from which a global object may or may not emerge.
Pairwise comparison theory offers one of the clearest mathematical incarnations of this principle: from Thurstone's law of comparative judgment to Saaty's analytic hierarchy process and its later developments, the primary data are relational and reciprocal rather than absolute \cite{Thurstone1927,Saaty1980,Saaty2008,Kulakowski2020}.

From this point of view, pairwise comparison matrices should not be regarded merely as practical devices for ranking.
They already encode a minimal relational ontology: what is given is not the intrinsic value of an item, but its relative position with respect to another one.
This perspective has generated a substantial mathematical literature on the extraction of weights, the reconstruction of priorities, and the treatment of incomplete comparative data \cite{CrawfordWilliams1985,Fichtner1986,Harker1987,IshizakaLusti2006,FedrizziGiove2007,BozokiFulopRonyai2010,Dijkstra2013}.
Such works show that even at the classical level, pairwise comparisons are not static numerical arrays but structured local data subject to reconstruction, optimization, and compatibility constraints.

A second central issue is inconsistency.
If local comparative judgments are primary, then coherence is no longer automatic; it becomes a problem.
The classical triangular obstructions measured by inconsistency indices are therefore not merely numerical imperfections but signals that a global comparative structure may fail to exist, or may exist only after correction.
This question has been studied from several complementary viewpoints, including axiomatizations of inconsistency indices, rankings of inconsistency, numerical investigations, completion procedures, and reduction algorithms \cite{Koczkodaj1993,BozokiRapcsak2008,BrunelliCanalFedrizzi2013,BrunelliFedrizzi2015,Brunelli2017,Brunelli2018,Csato2018,MagnotMazurekCernanova2023}.
In more recent work, random and geometric approaches have suggested that inconsistency itself should be treated as a structured observable, rather than as a purely exogenous perturbation \cite{Magnot2024}.

The present article starts from the idea that pairwise comparison data admit a more geometric and more foundational reading than is usually made explicit.
More precisely, when the coefficient set is replaced by a general group \(H\), reciprocal pairwise comparison matrices define a natural group \(\PC_n(H)\), and this simple algebraic observation already shifts the picture.
The data are no longer only comparative; they become group-valued and therefore compatible with constructions based on orientation, inversion, and ordered composition.
This bridge between pairwise comparisons and group-valued structures was already indicated in previous work, where links with discretized gauge-theoretic ideas were explored \cite{Magnot2018,Magnot2019}, linked with \cite{RovelliVidotto2015}.
The question then becomes: can one organize such comparative data in a way analogous to the organization of local geometric or gauge-type data?

This question is natural from the viewpoint of mathematical physics.
Lattice gauge theory, simplicial gauge theory, and related discrete formalisms assign group-valued data to local geometric pieces and recover coarser or global structures by composition rules compatible with subdivision \cite{Wilson1974,Kogut1979,ChristiansenHalvorsen2012}.
At a more conceptual level, topological quantum field theory, Chern--Simons theory, and higher gauge theory have taught us that orientation, gluing, and functoriality are not secondary technicalities but structural principles through which local information acquires global meaning \cite{Atiyah1988,Freed1995,BaezSchreiber2007}linked with scle problems \cite{Hossenfelder2013,Hossenfelder2021,Hossenfelder2022,RovelliVidotto2015}.
The present work does not identify pairwise comparisons with gauge theory.
Rather, it takes seriously the idea that comparative data, once endowed with an appropriate group-valued structure, may be organized by a similar local-to-global syntax.

This leads naturally to the notion of a \(\PC_n(H)\)-valued cosurface.
The term is meant to emphasize that the relevant data are attached not to points but to oriented pieces that admit decomposition and reversal.
In the stochastic setting, this perspective resonates with earlier works on Lie-group-valued stochastic measures, Yang--Mills fields on surfaces, and Markovian holonomy fields, where local group-valued data are organized through gluing and probabilistic consistency \cite{AlbeverioHoeghKrohnHolden1986,Sengupta1992,Sengupta2008,Levy2003,Levy2010}.
It also connects with the study of stochastic cosurfaces and their relation to topological field-theoretic structures \cite{Magnot2022}.
What is new here is that the group-valued labels are no longer primarily interpreted as holonomies of a gauge field, but as reciprocal comparative data.

A central philosophical point of the paper is that the global object should not be postulated first.
Instead, it should emerge from compatibility across scales.
For that reason, we begin with directed families of finite oriented discretizations.
At each scale, one obtains a finite configuration space of \(\PC_n(H)\)-valued assignments, and between scales one has coarse-graining maps induced by ordered refinements.
The corresponding projective limit is then the universal object of all compatible finite-level comparative configurations.
This projective viewpoint is not an auxiliary technical device; it is the mathematical expression of a foundational stance according to which a global relational object is accessible only through coherent families of partial descriptions.

Such a stance naturally calls for a cylindrical and projective probabilistic language.
Projective limits of probability spaces and cylindrical measures provide a classical framework for handling stochastic systems defined through compatible finite-level marginals \cite{Rao1971,Schwartz1973,Yamasaki1985}.
In our setting, a stochastic \(\PC_n(H)\)-valued cosurface is first of all a projectively compatible family of laws on finite configuration spaces, and only secondarily, when extension is available, a probability measure on the projective limit itself.
This means that stochasticity enters the theory at the same foundational level as refinement and compatibility: what is random is not an absolute state given once and for all, but a family of relational labels observed through compatible finite resolutions.

The comparison with other generalized stochastic frameworks is instructive.
In white-noise analysis and related theories, one often works with generalized random distributions, chaos expansions, and dual pairs of test and distribution spaces; Lévy-driven space-time white-noise equations provide a typical example of this distributional viewpoint \cite{LokkaOksendalProske2004}.
Our framework is of a different nature: it is not distributional in origin, but geometric and projective.
Likewise, in Glauber-type and birth--death dynamics on configuration spaces, randomness acts on the configuration itself through creation and annihilation mechanisms \cite{KondratievLytvynov2003,FinkelshteinKondratievKutoviy2011}.
Here, by contrast, the refinement support is fixed, and the randomness acts on comparative labels attached to oriented pieces.
Finally, recent work on random-time changes and fractional kinetics suggests another possible direction: once time-dependent compatible laws are introduced, one may envisage subordinated evolutions of pairwise-comparison-valued cosurfaces \cite{KochubeiKondratievDaSilva2020,Kondratiev2022}.
These comparisons help locate the present theory: it is neither a generalized process in the strict sense of white-noise analysis nor a branching process on evolving configuration spaces, but rather a projective stochastic theory of relational labels.

Another motivation comes from the conceptual role of inconsistency.
In the standard theory of pairwise comparisons, inconsistency is often treated as a quantity to minimize.
In the present framework, it acquires a more structural meaning.
Triangular defects become observables attached to oriented pieces and to their coarse-grained counterparts.
They can therefore be interpreted as discrete curvature-type quantities measuring the obstruction to the existence of a globally coherent comparative potential.
This interpretation does not replace the classical inconsistency literature; it reframes it.
The issue is no longer only to quantify deviation from consistency, but to understand how defects are localized, propagated, and transformed across scales.

The projective viewpoint also has a broader conceptual significance.
A global state space is often taken as primitive in both physics and applied mathematics.
Here, however, the primitive notion is rather that of compatible local relational data.
The global object is reconstructed only indirectly, through the universal property of the projective limit.
This makes the framework closer in spirit to relational and multi-scale approaches than to theories based on fully given absolute observables.
It also suggests that pairwise-comparison-valued cosurfaces may serve as a simple but nontrivial laboratory for thinking about relational structures in a mathematically controlled form.

The present article therefore pursues two goals.
The first is mathematical: to introduce the group \(\PC_n(H)\), define \(\PC_n(H)\)-valued cosurfaces, construct the associated projective system of finite-scale configurations, and develop the corresponding stochastic and cylindrical semantics.
The second is conceptual: to show that local comparative information can be organized in a way that is simultaneously geometric, algebraic, and probabilistic, and that this organization offers a natural foundational language for multi-scale relational structures.
In this sense, our aim is not to force one privileged semantic interpretation, but to isolate a common mathematical framework underlying several possible readings.

The paper is organized as follows.
In Section~2 we define the pairwise comparison group \(\PC_n(H)\) and study its coherence defects.
Section~3 introduces \(\PC_n(H)\)-valued cosurfaces and their compatibility with orientation reversal and ordered gluing.
In Section~4 we construct the finite-scale configuration spaces and the associated projective system.
Section~5 develops the stochastic semantics in terms of projectively compatible probability laws, cylinder sets, and cylindrical observables.
Section~6 is devoted to inconsistency observables and to the transversal geometric, algebraic, probabilistic, and causal interpretations of the construction.
Finally, the outlook discusses possible extensions toward generalized stochastic processes, time-dependent and subordinated dynamics, and broader semantic regimes, including causal and generalized probabilistic viewpoints \cite{MilzSakuldeePollockModi2020}.

To summarize the guiding intuition in one sentence: pairwise-comparison-valued cosurfaces provide a framework in which one begins with local comparisons rather than absolute observables, and in which global structure appears only as the projective coherence of relational data across scales.
\section{The pairwise comparison group \(\PC_n(H)\)}

\subsection{Reciprocal \(H\)-valued pairwise comparison matrices}

Let \(H\) be a group, written multiplicatively, with neutral element \(e\).
Fix \(n\geq 2\), and let
\[
U_n:=\{(i,j)\in \{1,\dots,n\}^2\mid i<j\}.
\]
We write \(N=\# U_n=\frac{n(n-1)}2\).

\begin{definition}
The set of \(H\)-valued reciprocal pairwise comparison matrices of order \(n\) is
\[
\PC_n(H):=
\left\{
A=(a_{ij})_{1\leq i,j\leq n}\in M_n(H)
\;\middle|\;
a_{ii}=e,\quad a_{ji}=a_{ij}^{-1}\ \text{for all } i,j
\right\}.
\]
\end{definition}

Thus an element of \(\PC_n(H)\) is completely determined by its upper triangular coefficients \((a_{ij})_{(i,j)\in U_n}\).
In particular, there is a canonical bijection
\[
\Phi:\PC_n(H)\longrightarrow H^{U_n},
\qquad
\Phi(A)=(a_{ij})_{(i,j)\in U_n}.
\]

\begin{remark}
When \(H=\mathbb{R}_{+}^{\times}\), one recovers the usual multiplicative reciprocal pairwise comparison matrices.
The present framework simply replaces the coefficient group \(\mathbb{R}_{+}^{\times}\) by an arbitrary group \(H\), possibly noncommutative.
\end{remark}

\subsection{Covariant/contravariant composition law}

The refinement formalism developed later requires a natural composition law on \(\PC_n(H)\).
Since the upper triangular part determines the whole matrix, the most natural choice is to multiply coefficients covariantly on the upper triangular side and to recover the lower triangular side by reciprocity.

\begin{definition}
For \(A=(a_{ij})\) and \(B=(b_{ij})\) in \(\PC_n(H)\), define \(A\star B\in \PC_n(H)\) by
\[
(A\star B)_{ij}:=
\begin{cases}
a_{ij}b_{ij}, & i<j,\\[0.3em]
e, & i=j,\\[0.3em]
\bigl((A\star B)_{ji}\bigr)^{-1}, & i>j.
\end{cases}
\]
\end{definition}

Equivalently, for \(i>j\),
\[
(A\star B)_{ij}
=
(a_{ji}b_{ji})^{-1}
=
b_{ji}^{-1}a_{ji}^{-1}
=
b_{ij}a_{ij}.
\]
Hence the law is indeed covariant on the upper triangular coefficients and contravariant on the lower triangular ones.

\begin{remark}
The reversal of order on the lower triangular side is not an additional convention but a direct consequence of the reciprocity relation \(a_{ji}=a_{ij}^{-1}\).
\end{remark}

\subsection{Group structure}

\begin{proposition}
The law \(\star\) endows \(\PC_n(H)\) with a group structure.
\end{proposition}

\begin{proof}
Transport the product of \(H^{U_n}\) through the bijection \(\Phi\).
More precisely, if
\[
x=(x_{ij})_{(i,j)\in U_n},
\qquad
y=(y_{ij})_{(i,j)\in U_n}
\]
belong to \(H^{U_n}\), define
\[
(x\cdot y)_{ij}:=x_{ij}y_{ij}.
\]
Then \(H^{U_n}\) is a group under componentwise multiplication.
By construction,
\[
\Phi(A\star B)=\Phi(A)\cdot \Phi(B)
\]
for all \(A,B\in \PC_n(H)\).
Therefore \(\star\) is associative.

The neutral element is the matrix \(\mathbf{1}\in \PC_n(H)\) defined by
\[
\mathbf{1}_{ij}=e
\qquad
\text{for all }1\leq i,j\leq n.
\]
Indeed, for \(i<j\),
\[
(A\star \mathbf{1})_{ij}=a_{ij}e=a_{ij},
\qquad
(\mathbf{1}\star A)_{ij}=ea_{ij}=a_{ij},
\]
and the lower triangular part follows by reciprocity.

Finally, the inverse of \(A=(a_{ij})\in \PC_n(H)\) for the law \(\star\) is the matrix \(A^{-\star}\in \PC_n(H)\) given by
\[
(A^{-\star})_{ij}=
\begin{cases}
a_{ij}^{-1}, & i<j,\\[0.3em]
e, & i=j,\\[0.3em]
a_{ij}^{-1}, & i>j,
\end{cases}
\]
that is, more intrinsically,
\[
(A^{-\star})_{ij}=
\begin{cases}
a_{ij}^{-1}, & i<j,\\[0.3em]
e, & i=j,\\[0.3em]
a_{ij}^{-1}=(a_{ji}^{-1})^{-1}=a_{ji}, & i>j.
\end{cases}
\]
Indeed, for \(i<j\),
\[
(A\star A^{-\star})_{ij}=a_{ij}a_{ij}^{-1}=e,
\]
hence \(A\star A^{-\star}=\mathbf{1}\), and similarly \(A^{-\star}\star A=\mathbf{1}\).
\end{proof}

\begin{corollary}
The map
\[
\Phi:\bigl(\PC_n(H),\star\bigr)\longrightarrow H^{U_n}
\]
is an isomorphism of groups.
In particular,
\[
\PC_n(H)\cong H^{\frac{n(n-1)}2}.
\]
\end{corollary}

\begin{remark}
If \(H\) is abelian, then \(\star\) is simply the usual coefficientwise multiplication.
If \(H\) is noncommutative, the upper triangular coefficients still multiply componentwise, but the lower triangular part carries the opposite order induced by inversion.
\end{remark}

\subsection{Coherence and inconsistency defects}

The next notion plays the role of a local curvature or defect observable.

\begin{definition}
Let \(A=(a_{ij})\in \PC_n(H)\), and let \(i,j,k\) be pairwise distinct.
The \emph{triangular defect} associated with the ordered triple \((i,j,k)\) is
\[
\kappa_{ijk}(A):=a_{ij}a_{jk}a_{ki}\in H.
\]
\end{definition}

The quantity \(\kappa_{ijk}(A)\) measures the failure of multiplicative transitivity along the cycle \(i\to j\to k\to i\).
In particular, \(A\) is perfectly coherent on the triple \((i,j,k)\) if and only if \(\kappa_{ijk}(A)=e\).

\begin{lemma}
For every \(A\in \PC_n(H)\) and every pairwise distinct \(i,j,k\),
\[
\kappa_{ikj}(A)=\kappa_{ijk}(A)^{-1}.
\]
\end{lemma}

\begin{proof}
Using reciprocity,
\[
\kappa_{ikj}(A)=a_{ik}a_{kj}a_{ji}
=a_{ik}a_{jk}^{-1}a_{ij}^{-1}
=(a_{ij}a_{jk}a_{ki})^{-1}
=\kappa_{ijk}(A)^{-1}.
\]
\end{proof}

This suggests the following global notion.

\begin{definition}
A matrix \(A=(a_{ij})\in \PC_n(H)\) is called \emph{coherent} if
\[
\kappa_{ijk}(A)=e
\qquad
\text{for all pairwise distinct } i,j,k.
\]
\end{definition}

\begin{proposition}
For \(A=(a_{ij})\in \PC_n(H)\), the following assertions are equivalent.
\begin{enumerate}[label=\textup{(\roman*)}]
\item \(A\) is coherent.
\item There exist elements \(\lambda_1,\dots,\lambda_n\in H\) such that
\[
a_{ij}=\lambda_i\lambda_j^{-1}
\qquad
\text{for all }1\leq i,j\leq n.
\]
\end{enumerate}
\end{proposition}

\begin{proof}
Assume first that \(a_{ij}=\lambda_i\lambda_j^{-1}\) for all \(i,j\).
Then
\[
\kappa_{ijk}(A)
=
(\lambda_i\lambda_j^{-1})
(\lambda_j\lambda_k^{-1})
(\lambda_k\lambda_i^{-1})
=e,
\]
hence \(A\) is coherent.

Conversely, assume that \(A\) is coherent.
Fix an index \(r\in\{1,\dots,n\}\), and define
\[
\lambda_i:=a_{ir}
\qquad
\text{for } i=1,\dots,n.
\]
Since \(\kappa_{ijr}(A)=e\), we have
\[
a_{ij}a_{jr}a_{ri}=e.
\]
Using \(a_{ri}=a_{ir}^{-1}\), this becomes
\[
a_{ij}a_{jr}=a_{ir},
\]
hence
\[
a_{ij}=a_{ir}a_{jr}^{-1}=\lambda_i\lambda_j^{-1}.
\]
Therefore (ii) holds.
\end{proof}

\begin{remark}
The previous proposition shows that coherence means that the matrix \(A\) is generated by a system of group-valued ``potentials'' \((\lambda_i)\), and that the triangular defects \(\kappa_{ijk}\) are the basic obstructions to the existence of such a potential representation.
This interpretation will later motivate the reading of inconsistency as a discrete curvature-type phenomenon.
\end{remark}
\section{\(\PC_n(H)\)-valued cosurfaces}

\subsection{Oriented gluing systems}

In order to define cosurfaces with values in \(\PC_n(H)\), we first isolate the minimal geometric syntax needed for orientation reversal and ordered gluing.

\begin{definition}
An \emph{oriented gluing system} is a pair
\[
(\Sigma,\mathcal R),
\]
where:
\begin{enumerate}[label=\textup{(\roman*)}]
\item \(\Sigma\) is a set equipped with an involution
\[
\Sigma\to\Sigma,\qquad \sigma\mapsto \bar\sigma,
\]
called \emph{reversal of orientation}, and satisfying
\[
\bar{\bar{\sigma}}=\sigma
\qquad
\text{for all }\sigma\in\Sigma;
\]

\item \(\mathcal R\) is a set of admissible refinement relations of the form
\[
\sigma \leadsto (\sigma_1,\dots,\sigma_k),
\qquad
k\geq 1,
\]
with \(\sigma,\sigma_1,\dots,\sigma_k\in\Sigma\), satisfying:
\begin{enumerate}[label=\textup{(\alph*)}]
\item \emph{unit:}
\[
\sigma\leadsto(\sigma)
\qquad
\text{for all }\sigma\in\Sigma;
\]

\item \emph{compatibility with orientation:} if
\[
\sigma\leadsto(\sigma_1,\dots,\sigma_k),
\]
then
\[
\bar{\sigma}\leadsto(\bar{\sigma}_k,\dots,\bar{\sigma}_1);
\]

\item \emph{associativity of refinement:} if
\[
\sigma\leadsto(\sigma_1,\dots,\sigma_k)
\]
and for each \(i\in\{1,\dots,k\}\),
\[
\sigma_i\leadsto(\tau_{i,1},\dots,\tau_{i,m_i}),
\]
then
\[
\sigma\leadsto(\tau_{1,1},\dots,\tau_{1,m_1},\tau_{2,1},\dots,\tau_{k,m_k}).
\]
\end{enumerate}
\end{enumerate}
\end{definition}

The point of this definition is that the geometry is encoded only through oriented pieces and admissible ordered decompositions. No ambient manifold is required at this stage.

\begin{remark}
In geometric examples, \(\Sigma\) will typically be a family of oriented hypersurfaces, faces, or codimension-one cells, and a refinement relation
\[
\sigma\leadsto(\sigma_1,\dots,\sigma_k)
\]
will mean that \(\sigma\) is decomposed into the ordered family \((\sigma_1,\dots,\sigma_k)\).
The order is essential as soon as the coefficient group \(H\) is noncommutative.
\end{remark}

\subsection{Orientation reversal and inversion}

The compatibility between orientation reversal and the group law on \(\PC_n(H)\) is governed by the inversion map.

\begin{definition}
Let
\[
\iota:\PC_n(H)\to\PC_n(H),\qquad A\mapsto A^{-\star},
\]
be the inversion map in the group \((\PC_n(H),\star)\).
\end{definition}

The next elementary fact will be used repeatedly.

\begin{lemma}\label{lem:inverse-product}
For every \(k\geq 1\) and every \(A_1,\dots,A_k\in\PC_n(H)\),
\[
(A_1\star\cdots\star A_k)^{-\star}
=
A_k^{-\star}\star\cdots\star A_1^{-\star}.
\]
\end{lemma}

\begin{proof}
This is the usual inversion rule in any group, applied to \((\PC_n(H),\star)\).
The statement follows by induction on \(k\), the case \(k=2\) being
\[
(A\star B)^{-\star}=B^{-\star}\star A^{-\star}.
\qedhere
\]
\end{proof}

\subsection{\(\PC_n(H)\)-valued cosurfaces}

We now define the central object of this article.

\begin{definition}
Let \((\Sigma,\mathcal R)\) be an oriented gluing system.
A \emph{\(\PC_n(H)\)-valued cosurface} on \((\Sigma,\mathcal R)\) is a map
\[
C:\Sigma\to \PC_n(H)
\]
such that:
\begin{enumerate}[label=\textup{(\roman*)}]
\item \emph{orientation rule:}
\[
C(\bar{\sigma})=C(\sigma)^{-\star}
\qquad
\text{for all }\sigma\in\Sigma;
\]

\item \emph{multiplicativity under ordered gluing:} whenever
\[
\sigma\leadsto(\sigma_1,\dots,\sigma_k),
\]
one has
\[
C(\sigma)=C(\sigma_1)\star\cdots\star C(\sigma_k).
\]
\end{enumerate}
\end{definition}

\begin{remark}
A \(\PC_n(H)\)-valued cosurface may be viewed as a multiplicative assignment of pairwise comparison data to oriented pieces of a geometric decomposition.
The refinement rule says that the value on a coarse piece is recovered from the ordered product of the values on the finer pieces.
\end{remark}

\subsection{Compatibility of the axioms}

The two defining axioms are compatible with one another precisely because orientation reversal in the gluing system is reversed at the level of ordered products.

\begin{proposition}\label{prop:orientation-gluing-compatible}
Let \(C:\Sigma\to \PC_n(H)\) satisfy the multiplicativity rule.
Assume moreover that
\[
C(\bar{\sigma})=C(\sigma)^{-\star}
\qquad
\text{for all }\sigma\in\Sigma.
\]
Then for every admissible refinement
\[
\sigma\leadsto(\sigma_1,\dots,\sigma_k),
\]
one has
\[
C(\bar{\sigma})
=
C(\bar{\sigma}_k)\star\cdots\star C(\bar{\sigma}_1).
\]
\end{proposition}

\begin{proof}
By the compatibility of \(\mathcal R\) with orientation,
\[
\bar{\sigma}\leadsto(\bar{\sigma}_k,\dots,\bar{\sigma}_1).
\]
On the other hand,
\[
C(\bar{\sigma})
=
C(\sigma)^{-\star}
=
(C(\sigma_1)\star\cdots\star C(\sigma_k))^{-\star}
\]
by multiplicativity on \(\sigma\).
Applying Lemma~\ref{lem:inverse-product}, we obtain
\[
C(\bar{\sigma})
=
C(\sigma_k)^{-\star}\star\cdots\star C(\sigma_1)^{-\star}.
\]
Using the orientation rule again gives
\[
C(\bar{\sigma})
=
C(\bar{\sigma}_k)\star\cdots\star C(\bar{\sigma}_1),
\]
as required.
\end{proof}

Thus the order reversal built into the refinement syntax matches exactly the group-theoretic inversion on \(\PC_n(H)\).

\subsection{Elementary examples}

We record two basic classes of examples.

\begin{example}[Finite cell decompositions]
Let \(K\) be a finite cell decomposition of a space \(M\), and let \(\Sigma\) be the set of oriented codimension-one cells of \(K\).
For each oriented cell \(\sigma\), the opposite orientation gives another element \(\bar{\sigma}\in\Sigma\).

If \(K'\) is a subdivision of \(K\), each oriented codimension-one cell \(\sigma\) of \(K\) is decomposed into an ordered family of oriented codimension-one cells of \(K'\),
\[
\sigma\leadsto(\sigma_1,\dots,\sigma_k).
\]
This defines an oriented gluing system.
A \(\PC_n(H)\)-valued cosurface on this system is then an assignment of a pairwise comparison matrix to each oriented cell, compatible with orientation reversal and subdivision.
\end{example}

\begin{example}[Triangulated hypersurfaces]
Let \(M\) be an oriented manifold and let \(\Sigma\) be a finite family of oriented piecewise smooth hypersurfaces, each of which admits compatible triangulated refinements.
If
\[
\sigma\leadsto(\sigma_1,\dots,\sigma_k)
\]
means that \(\sigma\) is cut into an ordered family of oriented sub-hypersurfaces, then any \(\PC_n(H)\)-valued cosurface is a rule assigning comparison data to each oriented hypersurface in such a way that the data on a whole hypersurface are recovered by ordered multiplication of the data on its pieces.
\end{example}

\subsection{Coherence along a cosurface}

The triangular defects introduced in Section~2 induce local inconsistency observables on cosurfaces.

\begin{definition}
Let \(C:\Sigma\to\PC_n(H)\) be a \(\PC_n(H)\)-valued cosurface.
For every \(\sigma\in\Sigma\) and every triple of pairwise distinct indices \(i,j,k\), define
\[
\kappa_{ijk}^C(\sigma):=\kappa_{ijk}(C(\sigma))
=
C(\sigma)_{ij}\,C(\sigma)_{jk}\,C(\sigma)_{ki}.
\]
\end{definition}

\begin{definition}
A \(\PC_n(H)\)-valued cosurface \(C\) is said to be \emph{locally coherent} if for every \(\sigma\in\Sigma\) and every pairwise distinct \(i,j,k\),
\[
\kappa_{ijk}^C(\sigma)=e.
\]
\end{definition}

Thus local incoherence of pairwise comparisons becomes a field of defect observables carried by oriented geometric pieces.

\begin{remark}
At this stage, coherence is imposed piecewise.
Later, once configuration spaces and refinements are introduced, one may also study how the defects \(\kappa_{ijk}^C(\sigma)\) propagate across scales under coarse graining.
This is one of the reasons for viewing these quantities as curvature-type observables.
\end{remark}

\subsection{Generated cosurfaces}

The multiplicativity axiom implies that a cosurface is often determined by its values on elementary pieces.

\begin{definition}
A subset \(E\subset \Sigma\) is called \emph{generating} if every \(\sigma\in\Sigma\) can be obtained through a finite chain of admissible refinements from elements of \(E\) and their orientation reversals.
\end{definition}

\begin{proposition}
Let \(E\subset \Sigma\) be a generating subset.
Assume that a map
\[
c:E\to \PC_n(H)
\]
is given and extended to \(\bar E:=\{\bar{\sigma}\mid \sigma\in E\}\) by
\[
c(\bar{\sigma})=c(\sigma)^{-\star}.
\]
If the refinement relations among elements generated by \(E\) are compatible with the ordered product in \(\PC_n(H)\), then there exists at most one \(\PC_n(H)\)-valued cosurface \(C\) on \(\Sigma\) extending \(c\).
\end{proposition}

\begin{proof}
Every value \(C(\sigma)\) is obtained from values on elementary pieces by iterated use of the gluing rule.
Hence uniqueness follows immediately from the defining axioms.
The only possible obstruction is independence of the chosen refinement path, which is exactly the stated compatibility condition.
\end{proof}

\begin{remark}
This proposition is a uniqueness statement rather than a general existence theorem.
In practice, existence is guaranteed when the refinement system comes from a concrete geometric subdivision procedure and the prescribed elementary data satisfy the induced compatibility relations.
\end{remark}

\section{Refinements and projective systems}

\subsection{Directed families of finite discretizations}

In order to pass from local comparative data to a multi-scale structure, we now consider directed families of finite oriented discretizations.

\begin{definition}
Let \(I\) be a directed partially ordered set.
A \emph{directed family of finite oriented discretizations} consists of the following data:
\begin{enumerate}[label=\textup{(\roman*)}]
\item for each \(r\in I\), a finite set \(\Sigma_r\) equipped with an involution
\[
\Sigma_r\to\Sigma_r,\qquad \sigma\mapsto \bar{\sigma},
\]
called orientation reversal;

\item for each pair \(s,r\in I\) with \(s\succeq r\), a refinement map
\[
\rho_{s,r}:\Sigma_r\longrightarrow W(\Sigma_s),
\]
where \(W(\Sigma_s)\) denotes the set of finite words in \(\Sigma_s\), such that the following conditions hold:
\begin{enumerate}[label=\textup{(\alph*)}]
\item \emph{identity:}
\[
\rho_{r,r}(\sigma)=(\sigma)
\qquad
\text{for all } \sigma\in\Sigma_r;
\]

\item \emph{orientation compatibility:} if
\[
\rho_{s,r}(\sigma)=(\sigma_1,\dots,\sigma_k),
\]
then
\[
\rho_{s,r}(\bar{\sigma})=(\bar{\sigma}_k,\dots,\bar{\sigma}_1);
\]

\item \emph{transitivity of refinement:} if \(t\succeq s\succeq r\), then
\[
\rho_{t,r}=\rho_{t,s}^{\#}\circ \rho_{s,r},
\]
where
\[
\rho_{t,s}^{\#}:W(\Sigma_s)\to W(\Sigma_t)
\]
is the map obtained by refining each letter and concatenating the resulting words.
\end{enumerate}
\end{enumerate}
\end{definition}

Let us make the notation \(\rho_{t,s}^{\#}\) explicit.
If
\[
w=(\sigma_1,\dots,\sigma_k)\in W(\Sigma_s)
\]
and
\[
\rho_{t,s}(\sigma_i)=\bigl(\tau_{i,1},\dots,\tau_{i,m_i}\bigr),
\]
then
\[
\rho_{t,s}^{\#}(w)
=
(\tau_{1,1},\dots,\tau_{1,m_1},\tau_{2,1},\dots,\tau_{k,m_k}).
\]

\begin{remark}
This structure is the finite-scale version of the oriented gluing systems introduced in Section~3.
The refinement map \(\rho_{s,r}\) tells us how a coarse oriented piece at scale \(r\) is decomposed into an ordered family of finer oriented pieces at scale \(s\).
The order matters as soon as the coefficient group \(H\) is noncommutative.
\end{remark}

\subsection{Configuration spaces at each scale}

We now define the comparative data carried by a finite discretization.

\begin{definition}
Let \(r\in I\).
The \emph{configuration space at scale \(r\)} is
\[
X_r:=
\left\{
x:\Sigma_r\to \PC_n(H)
\;\middle|\;
x(\bar{\sigma})=x(\sigma)^{-\star}
\text{ for all }\sigma\in\Sigma_r
\right\}.
\]
\end{definition}

Thus a configuration assigns a reciprocal pairwise comparison matrix to each oriented piece, with the natural inversion rule under reversal of orientation.

\begin{remark}
At this stage, \(X_r\) is an elementary configuration space: it records the values on the finite oriented pieces of the discretization.
The value of a larger composite piece is then obtained by multiplying along a refinement word, exactly as in the cosurface formalism of Section~3.
\end{remark}

Since \(\Sigma_r\) is finite, each \(X_r\) is finite whenever \(H\) is finite.

\begin{lemma}\label{lem:Xr-decomposition}
Fix \(r\in I\), and let \(E_r\subset \Sigma_r\) be a subset containing exactly one element from each pair \(\{\sigma,\bar{\sigma}\}\).
Then the restriction map
\[
X_r\longrightarrow \PC_n(H)^{E_r},
\qquad
x\longmapsto x|_{E_r},
\]
is a bijection.
\end{lemma}

\begin{proof}
Every \(x\in X_r\) is uniquely determined by its values on \(E_r\), because the values on the opposite orientations are prescribed by
\[
x(\bar{\sigma})=x(\sigma)^{-\star}.
\]
Conversely, any map \(E_r\to \PC_n(H)\) extends uniquely to an element of \(X_r\) by this rule.
\end{proof}

\begin{corollary}
For every \(r\in I\), the set \(X_r\) is naturally identified with a finite Cartesian power of \(\PC_n(H)\).
In particular, if \(H\) is finite, then \(X_r\) is finite.
\end{corollary}

\subsection{Word evaluation and coarse graining}

The refinement maps induce transition maps between configuration spaces by ordered multiplication along refined words.

\begin{definition}
Let \(r\in I\), \(x\in X_r\), and let
\[
w=(\sigma_1,\dots,\sigma_k)\in W(\Sigma_r)
\]
be a finite word.
The \emph{evaluation} of \(w\) under \(x\) is
\[
\langle x,w\rangle
:=
x(\sigma_1)\star\cdots\star x(\sigma_k)\in \PC_n(H).
\]
For the empty word, if needed, we set
\[
\langle x,\varnothing\rangle:=\mathbf{1},
\]
where \(\mathbf{1}\) denotes the neutral element of \(\PC_n(H)\).
\end{definition}

The compatibility with orientation is immediate.

\begin{lemma}\label{lem:word-orientation}
Let \(x\in X_r\), and let
\[
w=(\sigma_1,\dots,\sigma_k)\in W(\Sigma_r).
\]
Define the reversed oriented word
\[
\bar{w}:=(\bar{\sigma}_k,\dots,\bar{\sigma}_1).
\]
Then
\[
\langle x,\bar{w}\rangle=\langle x,w\rangle^{-\star}.
\]
\end{lemma}

\begin{proof}
Since \(x(\bar{\sigma}_i)=x(\sigma_i)^{-\star}\), we have
\[
\langle x,\bar{w}\rangle
=
x(\bar{\sigma}_k)\star\cdots\star x(\bar{\sigma}_1)
=
x(\sigma_k)^{-\star}\star\cdots\star x(\sigma_1)^{-\star}.
\]
By Lemma~\ref{lem:inverse-product}, this is exactly
\[
\bigl(x(\sigma_1)\star\cdots\star x(\sigma_k)\bigr)^{-\star}
=
\langle x,w\rangle^{-\star}.
\qedhere
\]
\end{proof}

We may now define the coarse-graining maps.

\begin{definition}
For \(s\succeq r\), define
\[
\Pi_{s,r}:X_s\to X_r
\]
by
\[
(\Pi_{s,r}x)(\sigma):=\langle x,\rho_{s,r}(\sigma)\rangle
\qquad
\text{for all }\sigma\in\Sigma_r.
\]
\end{definition}

\begin{proposition}\label{prop:coarse-graining-well-defined}
For every \(s\succeq r\), the map \(\Pi_{s,r}\) is well defined, that is,
\[
\Pi_{s,r}(X_s)\subset X_r.
\]
\end{proposition}

\begin{proof}
Let \(x\in X_s\), and let \(\sigma\in\Sigma_r\).
Write
\[
\rho_{s,r}(\sigma)=(\sigma_1,\dots,\sigma_k).
\]
By orientation compatibility of the refinement maps,
\[
\rho_{s,r}(\bar{\sigma})=(\bar{\sigma}_k,\dots,\bar{\sigma}_1).
\]
Therefore, by Lemma~\ref{lem:word-orientation},
\[
(\Pi_{s,r}x)(\bar{\sigma})
=
\langle x,\rho_{s,r}(\bar{\sigma})\rangle
=
\langle x,\rho_{s,r}(\sigma)\rangle^{-\star}
=
(\Pi_{s,r}x)(\sigma)^{-\star}.
\]
Hence \(\Pi_{s,r}x\in X_r\).
\end{proof}

\begin{remark}
The map \(\Pi_{s,r}\) is the basic coarse-graining operator of the theory: it sends a fine comparative configuration to the corresponding coarse one by ordered multiplication along the chosen refinement words.
\end{remark}

\subsection{Projective compatibility}

The family \((X_r,\Pi_{s,r})\) is projectively compatible.

\begin{theorem}\label{thm:projective-system}
The family of configuration spaces \((X_r)_{r\in I}\), together with the maps
\[
\Pi_{s,r}:X_s\to X_r
\qquad (s\succeq r),
\]
forms a projective system in the category of sets. More precisely:
\begin{enumerate}[label=\textup{(\roman*)}]
\item
\[
\Pi_{r,r}=\id_{X_r}
\qquad
\text{for all }r\in I;
\]

\item for all \(t\succeq s\succeq r\),
\[
\Pi_{t,r}=\Pi_{s,r}\circ \Pi_{t,s}.
\]
\end{enumerate}
\end{theorem}

\begin{proof}
Property \textup{(i)} follows from
\[
\rho_{r,r}(\sigma)=(\sigma),
\]
hence
\[
(\Pi_{r,r}x)(\sigma)=\langle x,(\sigma)\rangle=x(\sigma).
\]

For \textup{(ii)}, let \(x\in X_t\) and \(\sigma\in\Sigma_r\).
Write
\[
\rho_{s,r}(\sigma)=(\sigma_1,\dots,\sigma_k).
\]
Then
\[
(\Pi_{s,r}\circ\Pi_{t,s})(x)(\sigma)
=
\langle \Pi_{t,s}(x),\,(\sigma_1,\dots,\sigma_k)\rangle
=
(\Pi_{t,s}x)(\sigma_1)\star\cdots\star(\Pi_{t,s}x)(\sigma_k).
\]
By definition of \(\Pi_{t,s}\),
\[
(\Pi_{t,s}x)(\sigma_i)=\langle x,\rho_{t,s}(\sigma_i)\rangle.
\]
Thus
\[
(\Pi_{s,r}\circ\Pi_{t,s})(x)(\sigma)
=
\langle x,\rho_{t,s}(\sigma_1)\rangle
\star\cdots\star
\langle x,\rho_{t,s}(\sigma_k)\rangle.
\]
By associativity of \(\star\), this is exactly the evaluation of the concatenated word
\[
\rho_{t,s}^{\#}(\rho_{s,r}(\sigma)).
\]
Using the transitivity axiom
\[
\rho_{t,r}=\rho_{t,s}^{\#}\circ \rho_{s,r},
\]
we obtain
\[
(\Pi_{s,r}\circ\Pi_{t,s})(x)(\sigma)
=
\langle x,\rho_{t,r}(\sigma)\rangle
=
(\Pi_{t,r}x)(\sigma).
\]
Hence
\[
\Pi_{t,r}=\Pi_{s,r}\circ \Pi_{t,s}.
\qedhere
\]
\end{proof}

\subsection{The universal projective limit}

We may now define the limit object associated with the directed refinement system.

\begin{definition}
The \emph{projective limit} of the family \((X_r,\Pi_{s,r})\) is
\[
X_\infty
:=
\varprojlim_{r\in I} X_r
:=
\left\{
(x_r)_{r\in I}\in \prod_{r\in I} X_r
\;\middle|\;
\Pi_{s,r}(x_s)=x_r \text{ whenever } s\succeq r
\right\}.
\]
For each \(r\in I\), the canonical projection is denoted
\[
\pi_r:X_\infty\to X_r,\qquad \pi_r\bigl((x_t)_{t\in I}\bigr)=x_r.
\]
\end{definition}

An element of \(X_\infty\) is therefore a family of finite-scale comparative configurations compatible with all refinements.

\begin{proposition}
For every \(r\in I\) and every \(s\succeq r\),
\[
\Pi_{s,r}\circ \pi_s=\pi_r
\qquad
\text{on }X_\infty.
\]
\end{proposition}

\begin{proof}
This is immediate from the defining compatibility condition in \(X_\infty\).
\end{proof}

The next statement is the universal property that justifies the interpretation of \(X_\infty\) as the canonical object of compatible configurations across all scales.

\begin{theorem}\label{thm:universal-projective-limit}
The pair
\[
\bigl(X_\infty,(\pi_r)_{r\in I}\bigr)
\]
is a projective limit of the system \((X_r,\Pi_{s,r})\) in the category of sets.

Equivalently, if \(Y\) is a set and if one is given maps
\[
f_r:Y\to X_r
\qquad (r\in I)
\]
such that
\[
\Pi_{s,r}\circ f_s=f_r
\qquad
\text{for all } s\succeq r,
\]
then there exists a unique map
\[
f:Y\to X_\infty
\]
such that
\[
\pi_r\circ f=f_r
\qquad
\text{for all } r\in I.
\]
\end{theorem}

\begin{proof}
Define
\[
f:Y\to \prod_{r\in I} X_r,
\qquad
f(y):=(f_r(y))_{r\in I}.
\]
By the assumed compatibility,
\[
\Pi_{s,r}(f_s(y))=f_r(y)
\qquad
\text{for all }s\succeq r,
\]
hence \(f(y)\in X_\infty\).
Thus \(f\) indeed takes values in \(X_\infty\), and by construction
\[
\pi_r\circ f=f_r
\qquad
\text{for all }r.
\]

Uniqueness is immediate: if \(g:Y\to X_\infty\) also satisfies
\[
\pi_r\circ g=f_r
\qquad
\text{for all }r,
\]
then for every \(y\in Y\),
\[
g(y)=\bigl(\pi_r(g(y))\bigr)_{r\in I}
=
\bigl(f_r(y)\bigr)_{r\in I}
=
f(y).
\]
Hence \(g=f\).
\end{proof}

\begin{remark}
The object \(X_\infty\) should be viewed as the universal space of pairwise-comparison-valued configurations compatible with all scales of the discretization.
It is this object, rather than any single finite level \(X_r\), that encodes the full multi-scale content of the refinement system.
\end{remark}

\subsection{Relation with cosurfaces}

The projective system constructed above is an elementary configuration-level version of the cosurface formalism of Section~3.

Indeed, given \(x_r\in X_r\), one may evaluate \(x_r\) on any refinement word in \(W(\Sigma_r)\), and therefore on any composite piece built from the elementary oriented pieces of \(\Sigma_r\).
In this way, every finite-scale configuration determines a \(\PC_n(H)\)-valued cosurface on the oriented gluing system generated by \(\Sigma_r\) and its refinement words.
Conversely, every such cosurface restricts to an element of \(X_r\).

Thus the projective limit \(X_\infty\) may be interpreted as the universal object of finite-scale restrictions of compatible \(\PC_n(H)\)-valued cosurfaces across all discretization levels.

\subsection{A remark on group structures}

Although each set \(X_r\) is naturally identified with a finite Cartesian power of \(\PC_n(H)\), the transition maps \(\Pi_{s,r}\) are not, in general, homomorphisms for the pointwise group law on \(X_r\).
Indeed, the ordered multiplication defining coarse graining mixes the local factors in a nontrivial way.

Therefore, without further assumptions, the system \((X_r,\Pi_{s,r})\) is naturally a projective system in the category of sets, and later in the category of measurable or topological spaces.

A genuine projective system of groups reappears in special situations, for instance when the target group \(\PC_n(H)\) is abelian, or when all refinement words have length one.
For the general noncommutative multi-scale situation considered here, however, the set-theoretic projective formalism is the correct natural framework.

\section{Stochastic semantics}

\subsection{Measurable structures on the finite configuration spaces}

We now equip the finite-scale configuration spaces with probabilistic data.

For each \(r\in I\), let \(\mathcal F_r\) be a \(\sigma\)-algebra on \(X_r\) such that, for every \(s\succeq r\), the coarse-graining map
\[
\Pi_{s,r}:X_s\to X_r
\]
is \((\mathcal F_s,\mathcal F_r)\)-measurable.

\begin{remark}
When \(H\) is finite, each \(X_r\) is finite by Corollary~4.3, and the most natural choice is simply
\[
\mathcal F_r=\mathcal P(X_r).
\]
This is the main situation we have in mind for discrete stochastic models.
More generally, if \(H\) is endowed with a measurable structure, one may equip \(X_r\) with the induced product-type \(\sigma\)-algebra through the identification of Lemma~4.2.
\end{remark}

\begin{definition}
A \emph{probability law at scale \(r\)} is a probability measure
\[
\mathbb P_r:\mathcal F_r\to [0,1]
\]
on the measurable space \((X_r,\mathcal F_r)\).
\end{definition}

Thus a law at scale \(r\) describes a random assignment of pairwise-comparison data to the oriented pieces of the discretization \(\Sigma_r\).

\subsection{Compatibility under coarse graining}

The multi-scale nature of the construction is encoded by the requirement that finer laws project to coarser ones.

\begin{definition}
A family of probability laws
\[
(\mathbb P_r)_{r\in I}
\]
is called \emph{projectively compatible} if, for every \(s\succeq r\),
\[
(\Pi_{s,r})_{\ast}\mathbb P_s=\mathbb P_r.
\]
Equivalently,
\[
\mathbb P_r(A)=\mathbb P_s\bigl(\Pi_{s,r}^{-1}(A)\bigr)
\qquad
\text{for all }A\in\mathcal F_r.
\]
\end{definition}

This is the natural stochastic analogue of the compatibility condition defining the projective limit \(X_\infty\): the law seen at a coarse scale must coincide with the marginal of the law at any finer scale.

\begin{definition}
A \emph{stochastic \(\PC_n(H)\)-valued discretized cosurface system} is a projectively compatible family
\[
\bigl((X_r,\mathcal F_r,\mathbb P_r)\bigr)_{r\in I}.
\]
\end{definition}

\begin{remark}
At this stage, the stochastic object is not yet a single random point in a limiting space.
It is first given as a coherent family of finite-scale laws.
This point of view is especially natural when the limit object is defined projectively rather than as a concrete ambient space from the outset.
\end{remark}

\subsection{Cylinder sets and cylindrical observables on the limit space}

Let
\[
X_\infty=\varprojlim_{r\in I}X_r
\]
be the projective limit introduced in Section~4, and let
\[
\pi_r:X_\infty\to X_r
\]
denote the canonical projections.

\begin{definition}
A subset \(C\subset X_\infty\) is called a \emph{cylinder set} if there exist \(r\in I\) and \(A\in\mathcal F_r\) such that
\[
C=\pi_r^{-1}(A).
\]
The collection of all cylinder sets is denoted by
\[
\mathcal C_\infty.
\]
\end{definition}

\begin{lemma}
The family \(\mathcal C_\infty\) is an algebra of subsets of \(X_\infty\).
\end{lemma}

\begin{proof}
It is immediate that \(\varnothing\in\mathcal C_\infty\) and \(X_\infty\in\mathcal C_\infty\).
If \(C=\pi_r^{-1}(A)\), then
\[
X_\infty\setminus C=\pi_r^{-1}(X_r\setminus A)\in\mathcal C_\infty.
\]
Now let
\[
C_1=\pi_r^{-1}(A),\qquad C_2=\pi_s^{-1}(B),
\]
with \(A\in\mathcal F_r\) and \(B\in\mathcal F_s\).
Choose \(t\in I\) such that \(t\succeq r,s\).
Then
\[
C_1=\pi_t^{-1}\bigl(\Pi_{t,r}^{-1}(A)\bigr),
\qquad
C_2=\pi_t^{-1}\bigl(\Pi_{t,s}^{-1}(B)\bigr),
\]
hence
\[
C_1\cup C_2
=
\pi_t^{-1}\!\Bigl(\Pi_{t,r}^{-1}(A)\cup \Pi_{t,s}^{-1}(B)\Bigr)\in\mathcal C_\infty,
\]
and similarly for intersections.
\end{proof}

\begin{definition}
A function \(F:X_\infty\to \mathbb C\) is called a \emph{cylindrical observable} if there exist \(r\in I\) and a measurable function
\[
f_r:(X_r,\mathcal F_r)\to \mathbb C
\]
such that
\[
F=f_r\circ \pi_r.
\]
\end{definition}

\begin{remark}
Cylinder sets and cylindrical observables are the natural finite-resolution probes of the limit object.
They encode precisely the information accessible at a given discretization level.
In this sense, they provide the most intrinsic semantic layer on \(X_\infty\).
\end{remark}

\subsection{The cylindrical probability law on the projective limit}

A projectively compatible family of finite-scale laws determines canonically a finitely additive law on cylinder sets.

\begin{proposition}\label{prop:cylinder-premeasure}
Let
\[
(\mathbb P_r)_{r\in I}
\]
be a projectively compatible family of probability measures.
Define
\[
\mathbb P_\infty^{\mathrm{cyl}}:\mathcal C_\infty\to [0,1]
\]
by
\[
\mathbb P_\infty^{\mathrm{cyl}}\bigl(\pi_r^{-1}(A)\bigr):=\mathbb P_r(A),
\qquad
A\in\mathcal F_r.
\]
Then \(\mathbb P_\infty^{\mathrm{cyl}}\) is well defined and finitely additive on \(\mathcal C_\infty\).
\end{proposition}

\begin{proof}
We first prove that the definition is independent of the chosen representation of the cylinder set.

Suppose
\[
\pi_r^{-1}(A)=\pi_s^{-1}(B),
\qquad
A\in\mathcal F_r,\quad B\in\mathcal F_s.
\]
Choose \(t\in I\) such that \(t\succeq r,s\).
Then
\[
\pi_r^{-1}(A)
=
\pi_t^{-1}\bigl(\Pi_{t,r}^{-1}(A)\bigr),
\qquad
\pi_s^{-1}(B)
=
\pi_t^{-1}\bigl(\Pi_{t,s}^{-1}(B)\bigr).
\]
Since these cylinder sets are equal, their defining subsets in \(X_t\) must agree on the image of \(\pi_t\), and therefore
\[
\mathbb P_t\bigl(\Pi_{t,r}^{-1}(A)\bigr)
=
\mathbb P_t\bigl(\Pi_{t,s}^{-1}(B)\bigr).
\]
By projective compatibility,
\[
\mathbb P_t\bigl(\Pi_{t,r}^{-1}(A)\bigr)=\mathbb P_r(A),
\qquad
\mathbb P_t\bigl(\Pi_{t,s}^{-1}(B)\bigr)=\mathbb P_s(B),
\]
hence
\[
\mathbb P_r(A)=\mathbb P_s(B).
\]
So \(\mathbb P_\infty^{\mathrm{cyl}}\) is well defined.

Finite additivity follows from the fact that any finite family of cylinder sets may be represented at a common level \(t\), where additivity reduces to the additivity of the probability measure \(\mathbb P_t\).
\end{proof}

This cylindrical law is the minimal stochastic semantics carried by the compatible family \((\mathbb P_r)_{r\in I}\).

\begin{definition}
The map
\[
\mathbb P_\infty^{\mathrm{cyl}}:\mathcal C_\infty\to [0,1]
\]
constructed in Proposition~\ref{prop:cylinder-premeasure} is called the \emph{cylindrical stochastic semantics} associated with the compatible family \((\mathbb P_r)_{r\in I}\).
\end{definition}

\subsection{Extension to a genuine probability measure}

To obtain a genuine probability measure on the limit space, one needs an extension theorem from the cylinder algebra to the \(\sigma\)-algebra it generates.

Let
\[
\mathcal F_\infty:=\sigma(\mathcal C_\infty)
\]
be the cylinder \(\sigma\)-algebra on \(X_\infty\).

\begin{theorem}\label{thm:extension-limit-measure}
Assume that the cylindrical law \(\mathbb P_\infty^{\mathrm{cyl}}\) extends to a probability measure
\[
\mathbb P_\infty
\]
on \((X_\infty,\mathcal F_\infty)\).
Then, for every \(r\in I\),
\[
(\pi_r)_{\ast}\mathbb P_\infty=\mathbb P_r.
\]
\end{theorem}

\begin{proof}
Let \(A\in\mathcal F_r\).
By definition of \(\mathcal F_\infty\),
\[
\pi_r^{-1}(A)\in\mathcal C_\infty\subset \mathcal F_\infty.
\]
Hence
\[
(\pi_r)_{\ast}\mathbb P_\infty(A)
=
\mathbb P_\infty(\pi_r^{-1}(A)).
\]
Since \(\mathbb P_\infty\) extends \(\mathbb P_\infty^{\mathrm{cyl}}\),
\[
\mathbb P_\infty(\pi_r^{-1}(A))
=
\mathbb P_\infty^{\mathrm{cyl}}(\pi_r^{-1}(A))
=
\mathbb P_r(A).
\]
Therefore
\[
(\pi_r)_{\ast}\mathbb P_\infty=\mathbb P_r.
\qedhere
\]
\end{proof}

\begin{remark}
The existence of such an extension depends on the measure-theoretic setting.
In the finite case, or more generally under standard projective-limit hypotheses, this extension is classical; see for instance \cite{Rao1971,Schwartz1973,Yamasaki1985}.
For the purposes of the present article, the cylindrical semantics itself is already sufficient to encode the multi-scale stochastic content of the model.
\end{remark}

\subsection{Stochastic \(\PC_n(H)\)-valued cosurfaces}

We may now formulate the main stochastic notion associated with the projective system.

\begin{definition}
A \emph{stochastic \(\PC_n(H)\)-valued cosurface} on the directed refinement system \((\Sigma_r,\rho_{s,r})_{r\in I}\) is either:
\begin{enumerate}[label=\textup{(\roman*)}]
\item a projectively compatible family of probability laws
\[
(\mathbb P_r)_{r\in I}
\]
on the finite-scale configuration spaces \((X_r,\mathcal F_r)\); or equivalently,

\item when an extension exists, a probability measure
\[
\mathbb P_\infty
\]
on \((X_\infty,\mathcal F_\infty)\) whose marginals satisfy
\[
(\pi_r)_{\ast}\mathbb P_\infty=\mathbb P_r
\qquad
\text{for all }r\in I.
\]
\end{enumerate}
\end{definition}

\begin{remark}
The first formulation is intrinsic and does not require any additional extension theorem.
The second formulation is often more convenient conceptually, since it realizes the stochastic object as a law on the universal projective limit itself.
\end{remark}

\subsection{Expectation of cylindrical observables}

Cylinder observables admit a natural expectation with respect to the compatible family of finite-scale laws.

\begin{proposition}
Let
\[
F=f_r\circ \pi_r
\]
be a cylindrical observable on \(X_\infty\), where \(f_r:X_r\to\mathbb C\) is \(\mathcal F_r\)-measurable and integrable with respect to \(\mathbb P_r\).
Then the quantity
\[
\mathbb E_{\infty}^{\mathrm{cyl}}[F]
:=
\int_{X_r} f_r\, d\mathbb P_r
\]
is independent of the chosen representation of \(F\).
\end{proposition}

\begin{proof}
Assume that
\[
F=f_r\circ \pi_r=f_s\circ \pi_s
\]
for some \(r,s\in I\).
Choose \(t\in I\) with \(t\succeq r,s\).
Then
\[
F=(f_r\circ \Pi_{t,r})\circ \pi_t=(f_s\circ \Pi_{t,s})\circ \pi_t,
\]
hence
\[
f_r\circ \Pi_{t,r}=f_s\circ \Pi_{t,s}
\qquad
\mathbb P_t\text{-almost surely},
\]
and therefore
\[
\int_{X_t} f_r\circ \Pi_{t,r}\, d\mathbb P_t
=
\int_{X_t} f_s\circ \Pi_{t,s}\, d\mathbb P_t.
\]
Using projective compatibility,
\[
\int_{X_t} f_r\circ \Pi_{t,r}\, d\mathbb P_t
=
\int_{X_r} f_r\, d\mathbb P_r,
\qquad
\int_{X_t} f_s\circ \Pi_{t,s}\, d\mathbb P_t
=
\int_{X_s} f_s\, d\mathbb P_s.
\]
Hence
\[
\int_{X_r} f_r\, d\mathbb P_r
=
\int_{X_s} f_s\, d\mathbb P_s.
\qedhere
\]
\end{proof}

\begin{definition}
The map
\[
F\longmapsto \mathbb E_{\infty}^{\mathrm{cyl}}[F]
\]
defined on cylindrical observables is called the \emph{cylindrical expectation functional} associated with the stochastic \(\PC_n(H)\)-valued cosurface.
\end{definition}

\subsection{Interpretation}

The previous construction provides a stochastic semantics for pairwise-comparison-valued cosurfaces which is simultaneously local and multi-scale.

At each finite level \(r\), the law \(\mathbb P_r\) describes a random comparative configuration on the discretization \(\Sigma_r\).
Projective compatibility expresses the fact that the coarse information observed at scale \(r\) is exactly the marginal of the finer information available at scale \(s\).
The projective limit \(X_\infty\) then plays the role of a universal support for all compatible comparative configurations, while the cylindrical law \(\mathbb P_\infty^{\mathrm{cyl}}\) and the cylindrical expectation functional encode the stochastic content accessible through finite-resolution observations.

In this sense, a stochastic \(\PC_n(H)\)-valued cosurface is not merely a random matrix-valued field.
It is a coherent probabilistic organization of local comparative data across scales, compatible with orientation, ordered gluing, and refinement.
\section{Inconsistency observables and transversal interpretations}

\subsection{Local inconsistency observables at finite scale}

We now return to the triangular defects introduced in Section~2 and reinterpret them as observables on the finite-scale configuration spaces.

\begin{definition}
Fix \(r\in I\), an oriented piece \(\sigma\in\Sigma_r\), and a triple of pairwise distinct indices \(i,j,k\in\{1,\dots,n\}\).
The associated \emph{local inconsistency observable} is the map
\[
\mathcal K_{r,\sigma}^{ijk}:X_r\to H,
\qquad
\mathcal K_{r,\sigma}^{ijk}(x):=\kappa_{ijk}(x(\sigma)).
\]
\end{definition}

Thus \(\mathcal K_{r,\sigma}^{ijk}(x)\) measures the defect of multiplicative coherence of the pairwise-comparison matrix assigned by \(x\) to the oriented piece \(\sigma\).

\begin{lemma}
For every \(r\in I\), every \(\sigma\in\Sigma_r\), and every pairwise distinct \(i,j,k\),
\[
\mathcal K_{r,\bar{\sigma}}^{ijk}(x)
=
\bigl(\mathcal K_{r,\sigma}^{ikj}(x)\bigr)
\]
for all \(x\in X_r\).
In particular,
\[
\mathcal K_{r,\bar{\sigma}}^{ijk}(x)=e
\quad\Longleftrightarrow\quad
\mathcal K_{r,\sigma}^{ijk}(x)=e.
\]
\end{lemma}

\begin{proof}
Since \(x(\bar{\sigma})=x(\sigma)^{-\star}\), the coefficient relations of Section~2 imply that triangular defects for opposite orientations are inverted up to reversal of the cyclic order.
More precisely, by Lemma~2.6 applied to the matrix \(x(\sigma)\),
\[
\kappa_{ijk}(x(\bar{\sigma}))
=
\kappa_{ijk}(x(\sigma)^{-\star})
=
\kappa_{ikj}(x(\sigma)).
\]
This is exactly the stated identity.
The final equivalence follows because both quantities are equal to the neutral element simultaneously.
\end{proof}

\begin{definition}
Let \(r\in I\).
A configuration \(x\in X_r\) is called \emph{locally coherent at scale \(r\)} if
\[
\mathcal K_{r,\sigma}^{ijk}(x)=e
\]
for every \(\sigma\in\Sigma_r\) and every triple of pairwise distinct indices \(i,j,k\).
\end{definition}

This is the direct finite-scale analogue of the notion of coherent \(\PC_n(H)\)-valued cosurface introduced earlier.

\subsection{Coarse-grained inconsistency observables}

The refinement formalism allows us to measure inconsistency not only on elementary pieces, but also on coarse pieces seen through finer configurations.

\begin{definition}
Let \(s\succeq r\), let \(\sigma\in\Sigma_r\), and let \(i,j,k\) be pairwise distinct.
The \emph{coarse-grained inconsistency observable} at level \(s\) associated with \((r,\sigma,i,j,k)\) is the map
\[
\mathcal K_{s\to r,\sigma}^{ijk}:X_s\to H,
\qquad
\mathcal K_{s\to r,\sigma}^{ijk}(x)
:=
\kappa_{ijk}\bigl((\Pi_{s,r}x)(\sigma)\bigr).
\]
\end{definition}

By definition,
\[
\mathcal K_{s\to r,\sigma}^{ijk}
=
\mathcal K_{r,\sigma}^{ijk}\circ \Pi_{s,r}.
\]
Thus the inconsistency of a coarse piece may be read as an observable on finer-scale configurations.

\begin{remark}
In general, coarse-grained inconsistency is not a simple function of the individual local inconsistencies of the refined pieces.
Indeed, the ordered product in \(\PC_n(H)\) may create or cancel defect terms, especially in the noncommutative case.
This is precisely one of the reasons why the refinement formalism is nontrivial: coherence does not propagate in a purely pointwise manner.
\end{remark}

\subsection{Cylindrical inconsistency observables on the projective limit}

The local and coarse-grained inconsistency observables defined above naturally lift to the projective limit.

\begin{definition}
Fix \(r\in I\), \(\sigma\in\Sigma_r\), and pairwise distinct \(i,j,k\).
The corresponding \emph{cylindrical inconsistency observable} on \(X_\infty\) is
\[
\mathcal K_{\infty,\sigma}^{ijk}:X_\infty\to H,
\qquad
\mathcal K_{\infty,\sigma}^{ijk}
:=
\mathcal K_{r,\sigma}^{ijk}\circ \pi_r.
\]
\end{definition}

\begin{proposition}
The map \(\mathcal K_{\infty,\sigma}^{ijk}\) is a cylindrical observable on \(X_\infty\).
Moreover, if \(s\succeq r\), then
\[
\mathcal K_{\infty,\sigma}^{ijk}
=
\mathcal K_{s\to r,\sigma}^{ijk}\circ \pi_s.
\]
\end{proposition}

\begin{proof}
The first statement is immediate from the definition.
For the second one, let \(x_\infty\in X_\infty\).
Then
\[
\mathcal K_{s\to r,\sigma}^{ijk}(\pi_s(x_\infty))
=
\kappa_{ijk}\bigl((\Pi_{s,r}\pi_s(x_\infty))(\sigma)\bigr).
\]
Since \(\Pi_{s,r}\circ \pi_s=\pi_r\) on \(X_\infty\), this becomes
\[
\kappa_{ijk}\bigl((\pi_r(x_\infty))(\sigma)\bigr)
=
\mathcal K_{r,\sigma}^{ijk}(\pi_r(x_\infty))
=
\mathcal K_{\infty,\sigma}^{ijk}(x_\infty).
\]
Hence
\[
\mathcal K_{\infty,\sigma}^{ijk}
=
\mathcal K_{s\to r,\sigma}^{ijk}\circ \pi_s.
\qedhere
\]
\end{proof}

This result shows that inconsistency observables are genuinely projective objects: they may be probed at any sufficiently fine scale, while referring to a coarse geometric piece.

\begin{definition}
Let \((\mathbb P_r)_{r\in I}\) be a projectively compatible family of probability measures.
When \(H\) is endowed with a measurable structure and the observable \(\mathcal K_{\infty,\sigma}^{ijk}\) is measurable, its law under the stochastic cosurface is called the \emph{inconsistency law} associated with \((\sigma,i,j,k)\).
\end{definition}

Thus stochastic \(\PC_n(H)\)-valued cosurfaces carry not only comparative data, but also stochastic defect fields.

\subsection{Geometric interpretation}

The construction developed in the previous sections may first be read geometrically.

The basic geometric data are oriented pieces together with compatible refinement procedures.
A finite-scale configuration assigns to each oriented piece a reciprocal comparison matrix, and the orientation rule
\[
x(\bar{\sigma})=x(\sigma)^{-\star}
\]
plays the role of a reversal law.
Refinement words encode decomposition into ordered subpieces, and the coarse-graining maps are obtained by evaluating ordered products along these decompositions.

From this point of view, a \(\PC_n(H)\)-valued cosurface is a multiplicative geometric field carried by oriented codimension-one pieces.
The projective limit \(X_\infty\) may then be viewed as the universal space of all comparative fields compatible with all scales of the chosen discretization.
The inconsistency observables \(\mathcal K_{r,\sigma}^{ijk}\) or \(\mathcal K_{\infty,\sigma}^{ijk}\) play the role of local curvature-type quantities: they vanish precisely when the comparative data satisfy local multiplicative coherence.

This geometric reading is intentionally weak in its hypotheses.
It does not presuppose a smooth manifold, a bundle, or a connection.
What matters is only the syntax of oriented decomposition and ordered gluing.

\subsection{Algebraic interpretation}

The same construction admits an algebraic reading.

At the coefficient level, the central object is the group
\[
\PC_n(H),
\]
which is canonically isomorphic to \(H^{n(n-1)/2}\) but carries a matrix interpretation adapted to reciprocal comparative data.
The inversion law under orientation and the product law under refinement turn local comparison matrices into multiplicative algebraic labels.

At the scale of discretizations, the assignment
\[
r\longmapsto X_r,
\qquad
(s\succeq r)\longmapsto \Pi_{s,r}
\]
defines a contravariant system indexed by the refinement order.
Equivalently, the refinement syntax generates a diagram in the category of sets, and later in the category of measurable spaces.
The projective limit
\[
X_\infty=\varprojlim_r X_r
\]
is then the universal algebraic receptacle for compatible families of finite-scale configurations.

This point of view explains why the construction is not tied to one specific semantic regime.
The same syntax of refinement can be interpreted in purely algebraic terms before any probabilistic or causal layer is added.

\subsection{Probabilistic interpretation}

The probabilistic meaning of the construction was developed in Section~5, but it is worth stressing its conceptual content.

At each scale \(r\), a law \(\mathbb P_r\) describes a random comparative configuration.
The compatibility condition
\[
(\Pi_{s,r})_\ast \mathbb P_s=\mathbb P_r
\]
expresses the consistency of observations across scales: the random coarse configuration seen at level \(r\) is exactly the marginal of the finer-scale random configuration seen at level \(s\).

The projective limit \(X_\infty\), together with the cylindrical law
\[
\mathbb P_\infty^{\mathrm{cyl}},
\]
thus provides a stochastic semantics in which local comparative information is organized coherently across all discretizations.
The cylindrical observables are precisely those measurable probes that depend only on finitely many finite-scale degrees of freedom.

In this setting, the inconsistency observables become random variables, so that one may study not only whether local coherence holds, but also how defect distributions evolve under refinement and coarse graining.

\subsection{Causal interpretation}

Although no causal structure is imposed in the axioms, the refinement system admits a natural causal reading.

Indeed, if
\[
\rho_{s,r}(\sigma)=(\sigma_1,\dots,\sigma_k),
\]
then the coarse value associated with \(\sigma\) is obtained from the finer values associated with \(\sigma_1,\dots,\sigma_k\) through the ordered product in \(\PC_n(H)\).
This may be interpreted as a local mechanism
\[
\PC_n(H)^k \longrightarrow \PC_n(H),
\qquad
(A_1,\dots,A_k)\longmapsto A_1\star\cdots\star A_k,
\]
or, at the stochastic level, as a deterministic channel between random variables carried by different scales.

Thus the refinement relation can be read as a directed dependence from finer pieces to coarser pieces.
At this level of generality, the construction does not provide a full causal theory with interventions and latent variables.
However, it does isolate a natural multi-scale dependency structure, and this already places the model close to causal semantics in a broad sense.

In particular, coarse-grained inconsistency observables may be interpreted as effects induced by local comparative data propagating through the refinement mechanism.

\subsection{Transversal character of the construction}

One of the main features of the present framework is its transversal character.

First, it is transversal with respect to geometry: the formalism depends only on orientation and refinement, and therefore applies to triangulations, cell decompositions, piecewise smooth hypersurfaces, or any setting where ordered gluing makes sense.

Second, it is transversal with respect to semantics: the same refinement syntax can be interpreted deterministically, probabilistically, or causally, without changing the underlying mathematical object.

Third, it is transversal with respect to scale: the finite-level configuration spaces, the coarse-graining maps, and the projective limit organize local data and global compatibility within one unified structure.

Finally, it is transversal with respect to interpretation.
The matrices involved are simultaneously:
\begin{itemize}
\item geometric labels attached to oriented pieces,
\item algebraic elements of a group,
\item carriers of inconsistency defects,
\item and potential comparative observables in decision-type or learning-type settings.
\end{itemize}

For this reason, the proposed theory should not be viewed as belonging exclusively to one of these disciplines.
Rather, it provides a common mathematical layer underlying several types of local comparative models.

\subsection{Remarks toward comparative data and learning models}

Although the present article is primarily structural, the formalism suggests several directions in comparative data analysis and learning.

A first natural direction is hierarchical ranking.
In such a setting, local pieces may carry pairwise comparisons between alternatives or hypotheses, and refinements encode how fine comparative judgments combine into coarser decisions.
The projective framework then provides a mathematically controlled local-to-global passage.

A second direction concerns the aggregation of local or distributed annotations.
When several local sources provide comparative information on overlapping pieces or regions, one is led to study the stochastic propagation of comparison data and the resulting inconsistency fields.
The cylindrical semantics developed above provides a natural language for this purpose.

A third direction concerns regularization.
If a predictive system outputs \(\PC_n(H)\)-valued data on local pieces, then the defect observables
\[
\mathcal K_{r,\sigma}^{ijk}
\]
or their coarse-grained versions may serve as coherence penalties.
In this way, the present framework suggests multi-scale regularizers based on gluing consistency rather than solely on pointwise losses.

These remarks are only intended as indications.
Their role is not to shift the center of the paper toward applications, but rather to show that the mathematical structure developed here can already accommodate several families of local comparative models.

\subsection{Final synthesis}

We may summarize the overall picture as follows.

Starting from a directed family of finite oriented discretizations and a coefficient group \(H\), we construct:
\begin{enumerate}[label=\textup{(\roman*)}]
\item the pairwise comparison group \(\PC_n(H)\);
\item the finite-scale configuration spaces \(X_r\);
\item the coarse-graining maps \(\Pi_{s,r}\);
\item the universal projective limit \(X_\infty\);
\item compatible stochastic semantics through projective families of measures;
\item local, coarse-grained, and cylindrical inconsistency observables.
\end{enumerate}

This yields a multi-scale theory of pairwise-comparison-valued cosurfaces in which local relational data may be glued, propagated, and analyzed across scales.
The theory is geometric in its syntax, algebraic in its target, probabilistic in its stochastic realization, and open to causal or learning-theoretic interpretations.
It is precisely this coexistence of several compatible readings that gives the construction its transversal nature.
\section{Outlook}

The framework developed in this article suggests several directions for further investigation, both at the conceptual level and at the level of possible mathematical extensions.
Its main interest, in our view, lies in the fact that it organizes local relational data through a geometric syntax of orientation, refinement, and gluing, while remaining open to several semantic readings.
For this reason, the present work should be regarded less as a completed theory than as a foundational layer for a broader class of multi-scale relational models.

\subsection{Relation with generalized stochastic processes}

A first natural direction is to compare the present construction with established families of generalized stochastic models.

In white-noise analysis and in the theory of Kondratiev-type spaces, generalized stochastic processes are typically described as random distributions, often through chaos expansions and dual pairs of test and distribution spaces.
By contrast, the objects considered here are not distributional in origin.
They arise from finite configuration spaces of \(\PC_n(H)\)-valued comparative data, organized by refinement maps and projective compatibility.
Their primary structure is therefore geometric and combinatorial rather than functional-analytic.
Nevertheless, there is a genuine structural analogy: in both settings, the global stochastic object is reconstructed from compatible finite-level data, and cylindrical or projective viewpoints play a central role \cite{Rao1971,Schwartz1973,Yamasaki1985,LokkaOksendalProske2004}.

A second comparison concerns branching, birth--death, or Glauber-type processes on configuration spaces.
In such theories, randomness acts on the configuration itself, for instance through creation and annihilation mechanisms \cite{KondratievLytvynov2003,FinkelshteinKondratievKutoviy2011}.
In the present framework, the geometric support is fixed at the level of the refinement system, while the randomness acts on the relational labels carried by oriented pieces.
In this sense, the proposed model is closer to a projective stochastic theory of comparative fields than to a stochastic population dynamics.

These comparisons suggest that pairwise-comparison-valued cosurfaces should be viewed as a distinct class of generalized stochastic objects: not generalized processes in the strict sense of white-noise analysis, and not branching processes on evolving configuration spaces, but rather projective stochastic systems of relational data on a refined geometric support.

\subsection{Toward time-dependent and subordinated dynamics}

The present article is essentially concerned with spatial and multi-scale organization.
No intrinsic time evolution has been imposed.
A natural extension would therefore consist in introducing a time-dependent family of compatible laws
\[
(\mathbb P_{r,t})_{r\in I,\; t\geq 0},
\]
or, equivalently, a time-indexed stochastic process with values in the projective limit \(X_\infty\), whenever such a realization exists.

This would open the way to a kinetic theory of \(\PC_n(H)\)-valued cosurfaces, where refinement compatibility is preserved at every time.
Once such a dynamics is available, one may further consider random time changes and subordinated evolutions, in the spirit of fractional kinetic models \cite{KochubeiKondratievDaSilva2020,Kondratiev2022}.
From that viewpoint, the present framework may serve as a spatial and relational substrate for future time-dependent theories, including non-Markovian or fractional extensions.

\subsection{Further semantic regimes}

Another direction concerns the enlargement of the semantic setting.
In the present work, the main stochastic semantics is classical and cylindrical.
However, the refinement syntax itself is independent of this particular choice.
This suggests the possibility of interpreting the same projective structure in broader probabilistic or causal frameworks, and perhaps eventually in quantum or generalized probabilistic settings \cite{MilzSakuldeePollockModi2020}.

Such extensions would require additional structure, and they are not pursued here.
Still, the current construction indicates that the passage from local comparative data to compatible multi-scale objects is not tied to one unique semantic regime.
This semantic flexibility is, in our opinion, one of the most significant conceptual features of pairwise-comparison-valued cosurfaces.

\subsection{Foundational perspective}

From a foundational point of view, the main lesson of the present work is that comparative information can be organized without starting from absolute observables.
Instead, one may begin with local relational data, impose compatibility under orientation reversal and gluing, and reconstruct a global object only projectively.
The resulting theory is therefore closer in spirit to a relational and multi-scale ontology than to a theory of globally given states.

In this respect, the inconsistency observables introduced above play a particularly important role.
They do not merely quantify numerical defects; they measure obstructions to global coherence.
This gives them an interpretation akin to curvature-type quantities for relational structures.
Such a viewpoint may be useful well beyond the specific pairwise-comparison setting considered here.

\subsection{Comparative data and learning-theoretic perspectives}

Although the present article is not primarily application-driven, the framework also points toward several families of comparative data models.
Hierarchical ranking, aggregation of local annotations, and coherence regularization in learning systems all involve local relational information that must be propagated across scales.
The formalism developed here suggests that these apparently different situations may share a common mathematical substrate.

We do not pursue these directions in the present paper.
Their importance lies mainly in showing that the proposed construction is not isolated: it may serve as a common language for local comparative models arising in decision, stochastic geometry, and information processing.

\subsection{Final remarks}

The theory of \(\PC_n(H)\)-valued cosurfaces proposed here should therefore be understood as a first step toward a broader study of multi-scale relational systems.
Its geometric syntax, projective organization, and stochastic semantics place it at the crossroads of several traditions, while preserving a clear internal identity.
Whether this framework ultimately develops toward generalized stochastic analysis, time-dependent relational dynamics, or more explicitly physical or informational models remains an open question.
What seems already clear, however, is that pairwise-comparison-valued cosurfaces provide a mathematically coherent way to think about local relational data beyond the level of isolated matrices.

\vskip 12pt

\paragraph{\bf Conflict of interest statement:} The author declares no conflict of interest.

\vskip 12pt

\paragraph{\bf Acknowledgements:} J.-P.M  thanks the France 2030 framework programme Centre Henri Lebesgue ANR-11-LABX-0020-01 
for creating an attractive mathematical environment.

\vskip 12pt

\paragraph{\bf Declaration of generative AI and AI-assisted technologies in the writing process}

During the preparation of this work the author used ChatGPT and Mistral AI in order to smoothen the expression in English. After using this tool/service, the author reviewed and edited the content as needed and takes full responsibility for the content of the publication.

\end{document}